%% file: paper.tex
\documentclass{llncs}
\usepackage{xspace,amsmath,amssymb,stmaryrd,enumerate,color,ulem}
\usepackage{hyperref}
\usepackage{algorithm}
\usepackage{algpseudocode}
\usepackage[all]{xy}

\newcommand{\fv}{\mathit{fv}}
\newcommand{\Lnb}{\mathcal{L}_{\neq\bot}}
\newcommand{\U}{\mathcal{U}}
\renewcommand{\L}{\mathcal{L}}
\newcommand{\F}{\mathcal{F}}
\newcommand{\R}{\mathcal{R}}

\newcommand{\rank}{\mbox{\rm rank}}

\newcommand{\Var}{\mathit{Var}}
\newcommand{\sem}[1]{\llbracket #1 \rrbracket}

\newcommand{\None}{\mathit{None}}
\newcommand{\Some}{\mathit{Some}}
\newcommand{\Interval}{\mathit{Interval}}
\newcommand{\op}{\star}
\newcommand{\glb}{\bigsqcap{\!}_{\Delta}}
\newcommand{\leqR}{\preceq}

\newcommand{\Vf}{V'}

\newcommand{\ifit}{\textbf{if }}

\newcommand{\then}{\textbf{ then }}
\newcommand{\elseit}{\textbf{ else }}

\begin{document}

\title{Lattice based Least Fixed Point Logic}

\author{Piotr Filipiuk 
   \and Flemming Nielson 
   \and Hanne Riis Nielson 
   }
\institute{DTU Informatics, Richard Petersens Plads,Technical University of Denmark, DK-2800 Kongens Lyngby, Denmark \\
           \email{\{pifi,nielson,riis\}@imm.dtu.dk}}



\maketitle


\begin{abstract}
  As software systems become more complex, there is an increasing need
  for new static analyses. Thanks to the declarative style, logic
  programming is an attractive formalism for specifying them. However,
  prior work on using logic programming for static analysis focused on
  analyses defined over some powerset domain, which is quite
  limiting. In this paper we present a logic that lifts this
  restriction, called Lattice based Least Fixed Point Logic (LLFP),
  that allows interpretations over any complete lattice satisfying
  Ascending Chain Condition. The main theoretical contribution is a
  Moore Family result that guarantees that there always is a unique
  least solution for a given problem. Another contribution is the
  development of solving algorithm that computes the least model of
  LLFP formulae guaranteed by the Moore Family result.
\end{abstract}


\section{Introduction}
\label{sec:Introduction}
\input{introduction.tex}

\section{The problem}
\label{sec:motivation}
\input{motivation.tex}

\label{sec:intervals}
\input{intervals.tex}

\section{Syntax and Semantics}
\label{sec:alfp-lat}
\input{alfp-lat.tex}

\section{Moore family result for LLFP}
\label{sec:moore-family}
\input{moore-family.tex}

\section{The Algorithm}
\label{sec:algo}
\input{algo.tex}






\section{Conclusions and Future Work}
\label{sec:conclusions}
\input{conclusions.tex}

\bibliographystyle{abbrv}
\bibliography{paper}

\newpage
\appendix

\input{proof-partial-order-alfp-lat.tex}
\input{proof-complete-lattice-alfp-lat.tex}
\input{proof-moore-family.tex}

\end{document}

%% file: introduction.tex
Nowadays, we heavily rely on software systems. At the same time they
become bigger and more complex, and hence the number of potential
errors increases. In order to achieve more reliable systems, formal
verification techniques may be applied. A widely used verification
technique is static analysis, which reasons about system behavior
without executing it. It is performed statically at compile-time, and
it computes safe approximations of values or behaviors that may occur
at run-time. Static analysis is increasingly recognized as a
fundamental technique for program verification, bug detection,
compiler optimization and software understanding.

Unfortunately, developing new static analyses is difficult and
error-prone. In order to overcome that problem it is desirable to
implement prototypes of analyses that are easy to analyse for
complexity and correctness. Since analysis specifications are
generally written in a declarative style, logic programming presents
an attractive model for producing executable specifications of
analyses. Furthermore, thanks to the advances in logic programming,
the associated solvers became more efficient.

In this paper we present a framework that facilitates rapid
prototyping of new static analyses. The approach taken falls within
the Abstract Interpretation \cite{bib:cousot79,bib:cousot77}
framework, thus there always is a unique best solution to the analysis
problem considered. The framework consists of the Lattice based Least
Fixed Point Logic (LLFP) and the associated solver. The most prominent
feature of the LLFP logic is its interpretation over complete lattices
satisfying the Ascending Chain Condition, which makes it possible to
express sophisticated analyses in LLFP. The solver combines a
continuation passing style algorithm with propagation of differences,
and uses prefix trees as its main data structure. The applicability of
the framework is illustrated by presenting a specification of interval
analysis, which could not be specified using logics traditionally used
such as Datalog \cite{bib:datalog1,bib:datalog2} or ALFP
\cite{bib:ssforalfp}.

This paper is organized as follows. In Section \ref{sec:motivation} we
present the problem we want to solve and indicate a solution. Section
\ref{sec:alfp-lat} introduces syntax and semantics of the LLFP
logic. In Section \ref{sec:moore-family} we establish our main
theoretical contribution; namely a Moore Family result for
LLFP. Section \ref{sec:algo} describes the solving algorithm for
LLFP. We conclude and discuss future work in Section
\ref{sec:conclusions}.

%% file: motivation.tex
There is an immense body of work on using logic for specifying static
analyses. However, logics traditionally used have some limitations. To
illustrate the problem, let us briefly introduce the Alternation free
Least Fixed Point Logic (ALFP) and then try to devise the ALFP
specifications of two analyses: detection of signs and interval
analysis.

\paragraph{Alternation-free Least Fixed Point Logic.} 
Many static analyses can be succinctly expressed using
Alternation-free Least Fixed Point Logic (ALFP)
\cite{bib:ssforalfp}. The logic is a generalization of Datalog
\cite{bib:datalog1,bib:datalog2} and it has proved to have a number of
properties essential for specifying static analyses such as the
existence of a unique least model. The syntax of ALFP is given by
$$
\begin{array}{lcl}
  v & ::= & x \mid a \\
  pre & ::= &  R(v_1, \dots, v_k) \mid \neg R(v_1, \dots, v_k) \mid pre_1 \wedge pre_2  \\
  &  \mid  &  pre_1 \vee pre_2 \mid \exists x: pre
  \mid v_1 = v_2 \mid v_1 \neq v_2\\
  cl & ::= &  R(v_1, \dots, v_k) \mid {\bf 1} \mid cl_1 \wedge cl_2 \mid pre \Rightarrow cl \mid \forall x: cl
\end{array}
$$
where we write $a$ for constants, $x$ for analysis variables, $v$ for
values, $R$ for predicates, $pre$ for preconditions, and $cl$ for
clauses.  The clauses are interpreted over a universe $\U$ of
constants, $a \in \U$. The interpretation is given in terms of
satisfaction relations $(\rho, \sigma) \models pre$ and $(\rho,
\sigma) \models cl$ where $\rho$ is an interpretation of predicates,
and $\sigma$ is an interpretation of variables. The definition is
standard and hence omitted. 
Due to the use of negation, we impose a \textit{stratification}
condition similar to the one in Datalog
\cite{bib:datalog1,bib:datalog2}. This intuitively means that no
predicate depends on the negation of itself. We refer to
\cite{bib:ssforalfp} for more details.
\begin{example}
Using the notion of stratification we can define equality $E$ and
non-equality $N$ predicates as follows
\[
(\forall x : E(x,x)) \wedge (\forall x : \forall y : \neg E(x,y)
\Rightarrow N(x,y))
\]
The formula is stratified, since predicate $E$ is fully asserted
before it is negatively queried in the clause asserting predicate $N$.
\label{example:eq-neq-alfp}
\end{example}

\paragraph{Detection of signs analysis.}

Now, let us consider the ALFP formulation of the detection of signs
analysis. The analysis aims to determine for each program point and
each variable, the possible sign (negative, zero or positive) that the
variable may have whenever the execution reaches that point. In the
following, we use program graphs as representation of the program
under consideration \cite{bib:pmc}. Compared to the classical flow
graphs \cite{bib:kildall,bib:ppa}, the main difference is that in the
program graphs the actions label the edges rather than the
states. Here we focus on three types of actions: assignments, boolean
expressions and the skip action. For simplicity we assume that
assignments are in three-address form. The analysis is defined by
predicate $A$, and we begin with initializing the initial state of the
program graph, $q_0$, with all possible signs for all variables $v$
occurring in the underlying program graph
\[
\bigwedge_{v \in \Var} A(q_0,v,-) \wedge A(q_0,v,0) \wedge A(q_0,v,+)
\]
Intuitively it indicates that at state $q_0$ all variables may have
all possible values. Now we consider the ALFP specifications for each
type of action. Whenever we have $q_s \xrightarrow{x:=y \op z} q_t$ in
the program graph we generate
\begin{eqnarray*}
&\forall s : \forall s_y : \forall s_z : A(q_s, y, s_y) \wedge A(q_s, z, s_z)
\wedge R_{\op}(s_y, s_z, s) \Rightarrow A(q_t, x, s) \wedge\\
&\forall v : \forall s : v \neq x \wedge A(q_s, v, s) \Rightarrow A(q_t, v, s)
\end{eqnarray*}
where we assume that we have a relation for each type of arithmetic
operation, denoted by $R_{\op}$ in the above formula. The first
conjunct states that for all possible values $s$, $s_y$ and $s_z$, if
at state $q_s$ the signs of variables $y$ and $z$ are $s_y$ and $s_z$,
respectively, and the sign of the result of evaluating the arithmetic
operation $\op$ is $s$, then at state $q_t$ variable $x$ will have
sign $s$. The second conjunct expresses that for all variables $v$ and
signs $s$, if the variable is different than $x$ and at state $q_s$ it
has sign $s$, then it will have the same sign at state
$q_t$. Similarly, whenever we have $q_s \xrightarrow{e} q_t$ or $q_s
\xrightarrow{skip} q_t$ in the program graph, we generate a clause
\[
\forall v : \forall s : A(q_s, v, s) \Rightarrow A(q_t, v, s)
\]
The clause simply propagates the signs of all variables along the edge
of the program graph, without altering it. 

In a similar manner we could formulate other analyses such as pointer
analysis
\cite{bib:pabdd,DBLP:conf/pods/LamWLMACU05,DBLP:conf/aplas/WhaleyACL05,DBLP:conf/oopsla/BravenboerS09},
or classical data flow analyses \cite{bib:mf,bib:reps93}. In more
general terms, logics traditionally used, e.g.~Datalog and ALFP, can
be used for specifying analyses defined over a powerset
domain. However, as we show in the next paragraph, many interesting
analyses are defined over some mathematical structure such as a
complete lattice. Thus, let us now consider interval analysis as an
example of such an analysis.


%% file: intervals.tex
\paragraph{Interval analysis.}

The purpose of interval analysis is to determine for each program
point an interval containing possible values of variables whenever
that point is reached during run-time execution. The analysis results
can be used for Array Bound Analysis, which determines whether an
array index is always within the bounds of the array. If this is the
case, a run-time check can safely be eliminated, which makes code more
efficient.

We begin with defining the complete lattice $(\Interval,
\sqsubseteq_I)$ over which the analysis is defined. The underlying set
is
\[
\Interval={\bot} \cup \{ [z_1,z_2] \mid z_1 \leq z_2, z_1 \in Z \cup \{ -
\infty \}, z_2 \in Z \cup \{ \infty \} \}
\]
where $Z$ is a finite subset of integers, $Z \subseteq \mathbb{Z}$,
and the integer ordering $\leq$ on $\mathbb{Z}$ is extended to an
ordering on $Z' = Z\cup \{ -\infty,\infty \}$ by taking for all $z \in
Z$: $-\infty \leq z$, $z \leq \infty$ and $-\infty \leq \infty$. In
the above definition, $\bot$ denotes an empty interval, whereas $[z_1,
z_2]$ is the interval from $z_1$ to $z_2$ including the end points,
where $z_1, z_2 \in Z$. The interval $[ - \infty, \infty ]$ is
equivalent to the top element, $\top$. In the following we use $i$
to denote an interval from $\Interval$.
The partial ordering $\sqsubseteq_I$ in $\Interval$ uses operations
$\inf$ and $\sup$
\[
\begin{array}{rcl c rcl}
  \inf(i) & = & \left\{\begin{array}{ll}
      \infty & \text{if } i=\bot\\ 
      z_1 & \text{if } i=[z_1,z_2]
\end{array}\right. & \hspace{5mm} &
\sup(i) & = & \left\{\begin{array}{ll}
-\infty & \text{if } i=\bot\\ 
z_2 & \text{if } i=[z_1,z_2]
\end{array}\right. \\
\end{array}
\]
and is defined as
\[
i_1 \sqsubseteq_I i_2 \text{ iff } \inf(i_2) \leq \inf(i_1)
\wedge \sup(i_1) \leq \sup(i_2)
\]
The intuition behind the partial ordering $\sqsubseteq_I$ in $\Interval$
is that
\[
i_1 \sqsubseteq_I i_2 \Leftrightarrow \{ z \mid z \text{ belongs to
} i_1 \} \subseteq \{ z \mid z \text{ belongs to } i_2 \}
\]

Unfortunately, due to limited expressiveness of the ALFP logic (and
similarly Datalog), interval analysis can not be specified using these
formalisms. Hence, in the following section we present a solution for
that problem; namely we introduce LLFP logic.

%% file: alfp-lat.tex
In the previous section we briefly introduced ALFP logic, which is
interpreted over a finite universe of atoms; in this section we
present an extension of ALFP called LLFP allowing interpretations over
complete lattices satisfying Ascending Chain Condition. We also allow
function terms as arguments of relations. Since functions over the
universe $\U$ can be represented as relations, we do not consider them
here. Instead, we focus on functions over a complete lattice $\sem{f}:
\L^k \rightarrow \L$, and we restrict our attention to monotone
functions only. Recall that a function $\sem{f}: \L_1 \rightarrow
\L_2$ between partially ordered sets $\L_1 = (\L_1, \sqsubseteq_1)$
and $\L_2 = (\L_2, \sqsubseteq_2)$ is monotone if
\[
\forall l, l' \in \L_1 : l \sqsubseteq_1 l' \Rightarrow \sem{f}(l)
\sqsubseteq_2 \sem{f}(l')
\] Let us begin with introducing necessary definitions.
\begin{definition}
\label{def:complete-lattice}
A complete lattice $\L = (\L,\sqsubseteq) =
(\mathcal{L},\sqsubseteq,\bigsqcup,\bigsqcap,\bot,\top)$ is a
partially ordered set $(\mathcal{L},\sqsubseteq)$ such that all
subsets have least upper bounds as well as greatest lower bounds.
\end{definition}
A subset $Y \subseteq \L$ of a partially ordered set $\L =
(\L,\sqsubseteq)$ is a \emph{chain} if
\[
\forall l_1, l_2 \in Y : (l_1 \sqsubseteq l_2) \vee ( l_2 \sqsubseteq l_1)
\]
Hence, a chain is a (possibly empty) subset of $\L$ that is totally
ordered. A sequence $(l_n)_n$ of elements in $\L$ is an
\emph{ascending chain} if
\[
n < m \Rightarrow l_n \sqsubseteq l_m
\]
We say that a sequence $(l_n)_n$ \emph{eventually stabilises} if and
only if
\[
\exists n_0 \in \bbbn : \forall n \in \bbbn : n > n_0 \Rightarrow l_n
= l_{n_0}
\]
The partially ordered set $\L$ satisfies \emph{Ascending Chain
  Condition} if and only if all ascending chains eventually
stabilise. Essentially, the Ascending Chain Condition guarantees that
the least fixed point computation always terminates.
Due to the use of negation in the logic, we need to introduce a
complement operator, $\complement$, in the underlying complete
lattice. The only condition that we impose on the complement is
anti-monotonicity i.e.~$\forall l_1,l_2 \in \L: l_1 \sqsubseteq l_2
\Rightarrow \complement l_1 \sqsupseteq \complement l_2$, which is
necessary for establishing Moore Family result.
The following definition introduces the syntax of LLFP.
\begin{definition}\label{def:syntax-llfp-monotone}
  Given fixed countable and pairwise disjoint sets $\mathcal{X}$ and
  $\mathcal{Y}$ of variables, a non-empty and finite universe
  $\mathcal{U}$, a complete lattice satisfying Ascending Chain
  Condition $\L$, finite alphabets $\mathcal{R}$ and $\mathcal{F}$ of
  predicate and function symbols, respectively, we define the set of
  LLFP formulae (or clause sequences), $cls$, together with clauses,
  $cl$, preconditions, $pre$, terms $u$ and lattice terms $V$ and $\Vf$ by the
  grammar:
  \begin{center}
    \begin{tabular}{ l c l }
      $u$ & ::= & $x \mid a$\\
      $V$ & ::= & $Y \mid [u]$ \\
      $\Vf$ & ::= & $V \mid f(\vec{\Vf})$ \\
      $pre$ & ::= & $ R(\vec u; V) \mid \neg R(\vec u; V) \mid Y(u) \mid pre_1 \wedge pre_2 \mid  pre_1 \vee pre_2$\\
      & $ \mid$ & $ \exists x: pre \mid \exists Y: pre$ \\
      $cl$ & ::= & $ R(\vec u; \Vf) \mid {\bf 1} \mid cl_1 \wedge cl_2 \mid pre \Rightarrow cl \mid \forall x: cl \mid \forall Y: cl$ \\
      $cls$ & ::= & $ cl_1,\ldots,cl_s$
    \end{tabular}
  \end{center}
  Here $x \in \mathcal{X}$, $a\in\mathcal{U}$, $Y\in\mathcal{Y}$, $R
  \in \mathcal{R}$, $f \in \mathcal{F}$, and $s\geq 1$. Furthermore,
  $\vec u$ and $\vec{\Vf}$ abbreviate tuples $(u_1, \ldots, u_k)$ and
  $(\Vf_1, \ldots, \Vf_k)$ for some $k\geq 0$, respectively.
\end{definition}
We write $\fv(\cdot)$ for the set of free variables in the argument
$\cdot$. Occurrences of $R(\vec u;V)$ and $\neg R(\vec u;V)$ in
preconditions are called {\em positive}, resp.~{\em negative}, queries
and we require that $\fv(\vec u)\subseteq{\cal X}$ and
$\fv(V)\subseteq{\cal Y} \cup {\cal X}$; these variables are
\textit{defining} occurrences. Occurrences of $Y(u)$ in preconditions
must satisfy $Y\in{\cal Y}$ and $\fv(u)\subseteq {\cal X}$; $Y$ is an
\textit{applied} occurrence, $u$ is a defining occurrence. Clauses of
the form $R(\vec u;\Vf)$ are called \textit{assertions}; we require that
$\fv(\vec u)\subseteq{\cal X}$ and $\fv(\Vf)\subseteq{\cal Y} \cup {\cal
  X}$ and we note that these variables are applied occurrences. A
clause $cl$ satisfying these conditions together with
$\fv(cl)=\emptyset$ is said to be \textit{well-formed}; we are only
interested in clause sequences $cls$ consisting of well-formed
clauses.

In order to ensure desirable theoretical and pragmatic properties in
the presence of negation, we impose a notion of
\textit{stratification} similar to the one in Datalog
\cite{bib:datalog1,bib:datalog2}. Intuitively, stratification ensures
that a negative query is not performed until the predicate has been
fully asserted. This is important for ensuring that once a
precondition evaluates to true it will continue to be true even after
further assertions of predicates.
\begin{definition}
  The formula $cls = cl_1,\cdots, cl_s$ is stratified if there exists
  a function $\rank : \mathcal{R} \rightarrow \{0,\cdots,s\}$ such that
  for all $i = 1,\cdots, s$:
  \begin{itemize}
  \item $\rank(R) = i$ for every assertion $R$ in $cl_i$;
  \item $\rank(R) \le i$ for every positive query $R$ in $cl_i$; and
  \item $\rank(R) < i$ for every negative query $\neg R$ in $cl_i$.
  \end{itemize}
\label{def:stratification}
\end{definition}
The following example illustrates the use of negation in the LLFP
formula.
\begin{example}
  Similarly to Example \ref{example:eq-neq-alfp}, we can define
  equality $E$ and non-equality $N$ predicates in LLFP as follows
\[
(\forall x : E(x;[x])), (\forall x : \forall Y : \neg E(x;Y)
\Rightarrow N(x;Y))
\]
According to Definition \ref{def:stratification} the formula is
stratified, since predicate $E$ is fully asserted before it is
negatively queried in the clause asserting predicate $N$. As a result
we can dispense with an explicit treatment of $=$ and $\neq$ in the
development that follows. On the other hand the Definition
\ref{def:stratification} rules out
\[
(\forall x : \forall Y : \neg P(x;Y) \Rightarrow Q(x;Y)),
(\forall x : \forall Y : \neg Q(x;Y) \Rightarrow P(x;Y))
\]
\end{example}

To specify the semantics of LLFP we introduce the interpretations
$\varrho$, $\varsigma$ and $\zeta$ of predicate symbols,
variables and function symbols, respectively. Formally we have
$$
\begin{array}{rl}
  \varrho: & \prod_{k} \mathcal{R}_{/k} \rightarrow \mathcal{U}^k \rightarrow \L\\
  \varsigma: & (\mathcal{X} \rightarrow \mathcal{U}) \times
  (\mathcal{Y} \rightarrow \Lnb) \\
  \zeta: &
  \prod_{k} \F_{/k} \rightarrow \L^k \rightarrow \L
\end{array}
$$
In the above $\mathcal{R}_{/k}$ stands for a set of predicate symbols
of arity $k$, and $\mathcal{R}$ is a disjoint union of
$\mathcal{R}_{/k}$, hence $\mathcal{R}=\biguplus_{k}
\mathcal{R}_{/k}$.
Similarly, $\F_{/k}$ is a set of function symbols of
arity $k$ over the complete lattice $\L$.  The
set $\F$ is then defined as disjoint unions of
$\F_{/k}$; hence $\F=\biguplus_{k} \F_{/k}$.
The interpretation of variables from $\mathcal{X}$ is given by
$\sem{x}(\zeta, \varsigma)=\varsigma(x)$, where $\varsigma(x)$ is the
element from $\mathcal{U}$ bound to $x\in{\cal X}$. Analogously, the
interpretation of variables from $\mathcal{Y}$ is given by
$\sem{Y}(\zeta, \varsigma)=\varsigma(Y)$, where $\varsigma(Y)$ is the
element from $\Lnb = \L \setminus \{ \bot \}$ bound to
$Y\in\mathcal{Y}$. We do not allow variables from $\mathcal{Y}$ to be
mapped to $\bot$ in order to establish a relationship between ALFP and
LLFP in the case of powerset lattice, i.e.~$\mathcal{P}(\U)$, which we
briefly describe later.
In order to give the interpretation of $[u]$, we introduce a function
$\beta: \mathcal{U}\rightarrow \L$. The $\beta$ function is called a
\textit{representation function} and the idea is that $\beta$ maps a
value from the universe $\U$ to the \textit{best} property describing
it. For example in the case of a powerset lattice, $\beta$ could be
defined by $\beta(a) = \{ a \}$ for all $a \in \U$. Then the
interpretation is given by $\sem{[u]}(\zeta,
\varsigma)=\beta(\sem{u}(\zeta, \varsigma))$.
The interpretation of function terms is defined as $\sem{f(\vec
  \Vf)}(\zeta,\varsigma)=\zeta(f)(\sem{\vec
  \Vf}(\zeta,\varsigma))$. For the functions we require that
$\zeta(f): \L^k \rightarrow \L$ is monotone. The interpretation of
terms is generalized to sequences $\vec u$ of terms in a point-wise
manner by taking $\sem{a}(\zeta, \varsigma)=a$ for all $a\in {\cal
  U}$, thus $\sem{(u_1,\ldots, u_k)}(\zeta,
\varsigma)=(\sem{u_1}(\zeta, \varsigma), \ldots, \sem{u_k}(\zeta,
\varsigma))$. The interpretation of lattice terms $V$ (and $\Vf$) is
generalized to sequences $\vec{V}$ (and $\vec{\Vf}$) of lattice terms in
the similar way.

The satisfaction relations for preconditions $pre$, clauses $cl$ and
clause sequences $cls$ are specified by:
\[
(\varrho, \varsigma) \models_{\beta} pre,\quad
(\varrho, \zeta, \varsigma) \models_{\beta} cl\quad
\mathrm{and}\ (\varrho, \zeta, \varsigma) \models_{\beta}
cls
\]
The formal definition is given in Table \ref{ALFPStarsemantics}; here
$\varsigma[x\mapsto a]$ stands for the mapping that is as $\varsigma$
except that $x$ is mapped to $a$ and similarly $\varsigma[Y\mapsto l]$
stands for the mapping that is as $\varsigma$ except that $Y$ is
mapped to $l \in \Lnb$. 
\begin{table}
\caption{Semantics of LLFP}
\label{ALFPStarsemantics}
$$
\begin{array}{lllll}
  (\varrho, \varsigma) & \models_{\beta} & R(\vec u;V) &
  \underline{\texttt{iff}} & \varrho(R)(\varsigma(\vec u)) \sqsupseteq \varsigma(V) \\
  (\varrho, \varsigma) & \models_{\beta} & \neg R(\vec u;V) &
  \underline{\texttt{iff}} & \complement(\varrho(R)(\varsigma(\vec
  u))) \sqsupseteq \varsigma(V) \\
  (\varrho, \varsigma) & \models_{\beta} & Y(u) & \underline{\texttt{iff}} & \beta(\varsigma(u)) \sqsubseteq \varsigma(Y)\\
  (\varrho, \varsigma) & \models_{\beta} & pre_1 \wedge pre_2 & \underline{\texttt{iff}} & (\varrho, \varsigma) \models_{\beta} pre_1 \text{ and } (\varrho, \varsigma) \models_{\beta} pre_2 \\
  (\varrho, \varsigma) & \models_{\beta} & pre_1 \vee pre_2 & \underline{\texttt{iff}} & (\varrho, \varsigma) \models_{\beta} pre_1 \text{ or } (\varrho, \varsigma) \models_{\beta} pre_2
  \\
  (\varrho, \varsigma) & \models_{\beta} & \exists x: pre & \underline{\texttt{iff}} & (\varrho, \varsigma[x\mapsto a]) \models_{\beta} pre \text{ for some } a\in{\cal U}
  \\
  (\varrho, \varsigma) & \models_{\beta} & \exists Y: pre & \underline{\texttt{iff}} & (\varrho, \varsigma[Y\mapsto l]) \models_{\beta} pre \text{ for some } l\in\Lnb
  \\ \\
  (\varrho, \zeta, \varsigma) & \models_{\beta} & R(\vec
  u;V') & \underline{\texttt{iff}} & \varrho(R)(\sem{\vec u}(\zeta,
  \varsigma)) \sqsupseteq \sem{V'}(\zeta,
  \varsigma) \\
  (\varrho, \zeta, \varsigma) & \models_{\beta} & {\bf 1} & \underline{\texttt{iff}} & \texttt{true} \\
  (\varrho, \zeta, \varsigma) & \models_{\beta} & cl_1 \wedge cl_2 & \underline{\texttt{iff}} & (\varrho, \zeta, \varsigma) \models_{\beta} cl_1 \text{ and } (\varrho, \zeta, \varsigma) \models_{\beta} cl_2 \\
  (\varrho, \zeta, \varsigma) & \models_{\beta} & pre \Rightarrow cl & \underline{\texttt{iff}} & (\varrho, \zeta, \varsigma) \models_{\beta} cl \text{ whenever } (\varrho, \varsigma) \models_{\beta} pre \\
  (\varrho, \zeta, \varsigma) & \models_{\beta} & \forall
  x:cl & \underline{\texttt{iff}} & (\varrho, \zeta, \varsigma[x \mapsto a]) \models_{\beta} cl \text{ for all }a \in \mathcal{U} \\
  (\varrho, \zeta, \varsigma) & \models_{\beta} & \forall
  Y:cl & \underline{\texttt{iff}} & (\varrho, \zeta, \varsigma[Y \mapsto l]) \models_{\beta} cl \text{ for all }l \in \Lnb
  \\ \\
  (\varrho, \zeta, \varsigma) & \models_{\beta} &
  cl_1,\cdots, cl_s &  \underline{\texttt{iff}} & (\varrho,
  \zeta, \varsigma) \models_{\beta} cl_i \text{ for all } i, 1\leq i\leq s
\end{array}
$$
\end{table}

\paragraph{Relationship to ALFP.}

As reader may have already noticed, in the case the underlying
complete lattice is $\mathcal{P}(\U)$ the two logics are essentially
equivalent. More precisely, in the case of powerset lattice,
$\mathcal{P}(\U)$, function $\beta$ given by $\beta(a) = \{ a \}$ for
all $a \in \U$, and without function terms we can translate LLFP
formula into a corresponding ALFP one and vice versa. Intuitively, we
get the following correspondence between interpretations of relations
\[
\forall \vec{a}, b : ( (\vec{a},b) \in \rho(R) 
\Leftrightarrow
\varrho(R)(\vec{a}) \supseteq \{ b \} )
\]
The idea is that a relation $R$ in LLFP with interpretation
$\varrho(R)\in{\U}^k\rightarrow \mathcal{P}(\U)$ is replaced by a
relation in ALFP (also named $R$) with interpretation $\rho(R) \in
{\cal P}({\cal U}^{k+1})$. Note that if $\varrho(R)(\vec a) = \bot$
then $\rho(R)$ does not contain any tuples with $\vec a$ as the first
$k$ components.

\paragraph{Interval analysis in LLFP.}

Now let us give an LLFP specification of interval analysis. The
analysis is defined by the predicate $A$. Similarly to Datalog or
ALFP, the specification is defined over a universe $\U$, which in this
case is a set of all variables, $\Var$, appearing in the program as
well as states in the underlying program graph. In addition, the LLFP
logic allows interpretations over complete lattices satisfying
Ascending Chain Condition. Here we use the lattice $(\Interval,
\sqsubseteq_I)$, defined in Section \ref{sec:motivation}.

The specification consists of the initialization clauses and clauses
corresponding to three types of actions in the underlying program
graph. First, for the initial state, $q_0$, we initialize all
variables in the program graph with the $\top$ element, denoting that
they may have all possible values
\[
\bigwedge_{v \in \Var} A(q_0,v;\top)
\]
Furthermore, whenever we have $q_s \xrightarrow{x:=y \op z} q_t$ in
the program graph we generate
\begin{eqnarray*}
& \forall i_y: \forall i_z: A(q_s,y;i_y) \wedge A(q_s,z;i_z)
\Rightarrow A(q_t,x;f_{\op}(i_y,i_z)) \wedge\\
& \forall v: \forall i: v \neq x \wedge A(q_s,v;i) \Rightarrow A(q_t,v;i)
\end{eqnarray*}
The first conjunct updates the possible interval of values for the
assigned variable (in that case for variable $x$), with the result of
evaluating the arithmetic operation $y \op z$. The second conjunct
propagates the analysis information for all variables except variable
$x$ without altering it. Furthermore, whenever we have $q_s
\xrightarrow{e} q_t$ or $q_s \xrightarrow{skip} q_t$ in the program
graph, we generate a clause
\[
\forall v : \forall i : A(q_s, v, i) \Rightarrow A(q_t, v, i)
\]
which simply propagates the analysis information along the edge of the
program graph, without making any changes.



%% file: moore-family.tex
In this section we establish a Moore family result for LLFP that
guarantees that there always is a unique best solution for LLFP
clauses.
\begin{definition} A Moore family is a subset $Y$ of a complete
  lattice $\L=(\L,\sqsubseteq)$ that is closed under greatest lower
  bounds: $\forall Y' \subseteq Y: \bigsqcap Y' \in Y$.
\end{definition}
It follows that a Moore family always contains a least element,
$\bigsqcap Y$, and a greatest element, $\bigsqcap \emptyset$, which
equals the greatest element, $\top$, from $\L$; in particular, a Moore
family is never empty. The property is also called the model
intersection property, since whenever we take a {\it meet} of a number
of models we still get a model.

Assume $cls$ has the form $cl_1, \ldots, cl_s$, and let $\Delta = \{
\varrho: \prod_{k} \mathcal{R}_{/k} \rightarrow \mathcal{U}^k
\rightarrow \L \}$ denote the set of interpretations $\varrho$ of
predicate symbols in $\mathcal{R}$. We also define the lexicographical
ordering $\leqR$ such that $\varrho_1 \leqR \varrho_2$ if and only if
there is some $1 \leq j \leq s$. where $s$ is the order of the
formula, such that the following properties hold:
\begin{enumerate}[(a)]
\item $\varrho_1(R)=\varrho_2(R)$ for all $R \in \mathcal{R}$
  with $\rank(R)<j$, \label{itm:alfp-lat-ord-rank-less}
\item $\varrho_1(R) \sqsubseteq \varrho_2(R)$ for all $R \in
  \mathcal{R}$ with $\rank(R)=j$, \label{itm:alfp-lat-ord-rank-eq}
\item either $j=s$ or $\varrho_1(R) \sqsubset \varrho_2(R)$
  for at least one $R \in \mathcal{R}$ with $\rank(R)=j$. \label{itm:alfp-lat-ord-rank-s}
\end{enumerate}
We say that $\varrho_1(R) \sqsubseteq \varrho_2(R)$ if and only if
$\forall \vec{a}\in \U^k: \varrho_1(R)(\vec{a}) \sqsubseteq
\varrho_2(R)(\vec{a})$, where $k \geq 0$ is the arity of $R$. Notice
that in the case $s=1$, the above ordering coincides with lattice
ordering $\sqsubseteq$. Intuitively, the lexicographical ordering
$\leqR$ orders the relations strata by strata starting with the strata
$0$. It is essentially analogous to the lexicographical ordering on
strings, which is based on the alphabetical order of their characters.

\begin{lemma}\label{lemma:partial-order-alfp-lat}
$\leqR$ defines a partial order.
\end{lemma}
\begin{proof}
See Appendix \ref{proof:lemma:partial-order-alfp-lat}.
\end{proof}

Assume $cls$ has the form $cl_1, \cdots, cl_s$ where $cl_j$ is the
clause corresponding to stratum $j$, and let $\mathcal{R}_j$ denote
the set of all relation symbols $R$ defined in $cl_1, \cdots, cl_j$
taking $\mathcal{R}_0 = \emptyset$.  Let $M \subseteq \Delta$ denote a
set of assignments which map relation symbols to relations.

\begin{lemma}\label{lemma:complete-lattice-alfp-lat}
  $\Delta=(\Delta, \leqR)$ is a complete lattice with the
  greatest lower bound given by
\[
\left(\glb M\right)(R) = \lambda \vec{a} . \bigsqcap \left\{ {\varrho}(R)(\vec{a}) \mid \varrho \in
M_{\rank(R)} \right\}
\]
where
\[
M_j = \left\{ \varrho \in M \mid \forall R' \ \rank(R') < j:
\varrho(R') = \left( \glb M \right) \left(R'\right) \right\}
\]
\end{lemma}
\begin{proof}
See Appendix \ref{proof:lemma:complete-lattice-alfp-lat}
\end{proof}

Note that $\glb M$ is well defined by induction on $j$ observing that
$M_0=M$ and $M_j \subseteq M_{j-1}$.

\begin{proposition}\label{prop:moore-family-alfp-lat}
  Assume $cls$ is a stratified LLFP clause sequence, $\varsigma_0$ and
  $\zeta_0$ are interpretations of free variables and function symbols
  in $cls$, respectively. Furthermore, $\varrho_0$ is an
  interpretation of all relations of rank 0. Then $\{ \varrho \mid
  (\varrho, \zeta_0, \varsigma_0) \models_{\beta} cls \wedge \forall
  R: \rank(R) = 0 \Rightarrow \varrho_0(R) \sqsubseteq \varrho(R) \}$
  is a Moore family.
\end{proposition}

\begin{proof}
See Appendix \ref{proof:prop:moore-family-alfp-lat}
\end{proof}

The result ensures that the approach falls within the framework of
Abstract Interpretation \cite{bib:cousot77,bib:cousot79}; hence we can
be sure that there always is a single best solution for the analysis
problem under consideration, namely the one defined in Proposition
\ref{prop:moore-family-alfp-lat}.

%% file: algo.tex
In this section we present the algorithm for solving LLFP clause
sequences, which extends the differential worklist algorithm by
Nielson et al. \cite{bib:ssforalfp,bib:sssuite}. The algorithm
computes the relations in increasing order on their rank and therefore
the negations present no obstacles. It completely abandons a
worklist-like data structures, which are typical for most classical
iterative fixpoint algorithms
\cite{DBLP:journals/scp/FechtS99}. Instead, we adapt the recursive
topdown approach of Le Charlier and van Hentenryck
\cite{bib:LeCharlier92} which is enhanced by continuation based
semi-naive iteration \cite{bib:Balbin87,bib:FechtSeidl98}.


In the following we assume that prior to solving the LLFP formula, all
the clauses are transformed into a form such that all \textit{applied}
occurrences of variables $Y \in \mathcal{Y}$ in preconditions,
i.e.~$Y(u)$, are not followed by their \textit{defining} occurrences,
i.e.~$R(\vec u;Y)$ and $\neg R(\vec u;Y)$. This is necessary to
correctly perform late bindings of variables $Y \in \mathcal{Y}$ in
the presence of $Y(u)$ construct.

The algorithm operates with (intermediate) representations of the two
interpretations $\varsigma$ and $\varrho$ of the semantics; we
shall call them {\tt env} and {\tt result}, respectively, in the
following. The data structure {\tt env} is supplied as a parameter to
the functions of the algorithms, and it represents partial
environment. The data structure {\tt result} is an imperative data
structure that is updated as we progress.

The partial environment {\tt env} is implemented as a map from
variables to their optional values. In the case the variable is
undefined it is mapped into $\None$. Otherwise, depending on its type
it is mapped to $\Some(a)$ or $\Some(l)$, which means that the
variable is bound to $a \in{\cal U}$, or $l \in \Lnb$,
respectively. The main operation on {\tt env} is the function
\textsc{unify}, defined as follows
$$
\textsc{unify}(\beta, {\tt env},(\vec{u};V), (\vec{a};l)) = 
\left\{\begin{array}{ll}

\emptyset & \hbox{if } \textsc{unify}_\textsc{U} ({\tt env},
\vec{u}, \vec{a}) = \mbox{fail}\\

\textsc{unify}_\textsc{L}(\beta, {\tt env'}, V, l) & \hbox{if } \textsc{unify}_\textsc{U} ({\tt env},
\vec{u}, \vec{a}) = {\tt env'}\\
\end{array}\right.
$$
It uses two auxiliary functions that perform unifications on each
component of the relation. For the first component, which ranges over
the universe $\U$, the function is given by
$$
\textsc{unify}_\textsc{U}({\tt env},u,a) = 
\left\{\begin{array}{ll}

{\tt env} & \hbox{if } (u\in\mathcal{X}\wedge{\tt
  env}[u]=\Some(a))\vee u=a\\

{\tt env}[u\mapsto \Some(a)] &  \hbox{if }
u\in\mathcal{X}\wedge{\tt env}[u]=\None\\

\hbox{fail} & \hbox{otherwise}
\end{array}\right.
$$
It performs a unification of an argument $u$ with an element $a \in
\U$ in the environment {\tt env}. In the case when the unification
succeeds the modified environment is returned, otherwise the function
fails. The funcion is extended to $k$-tuples in a straightforward
way. The definition of the unification function for the lattice
component is given by
$$
\textsc{unify}_\textsc{L}(\beta, {\tt env}, V, l) = 
\left\{\begin{array}{l}

    \{ {\tt env}[ V \mapsto \Some(l \sqcap l_V) ] \} \\
    \hspace{5mm} 
    \hbox{if } 
    V\in\mathcal{Y}\wedge{\tt env}[V]=\Some(l_V)\wedge l \sqcap l_V \neq \bot \\

    \{ {\tt env}[V\mapsto \Some(l)] \} \\
    \hspace{5mm} 
    \hbox{if }
    V\in\mathcal{Y}\wedge{\tt env}[V]=\None \wedge l \neq \bot\\

    \{ {\tt env} \} \hbox{ if } V = [u] \wedge\\
    \hspace{5mm} 
    ((u\in\mathcal{X}\wedge{\tt env}[u]=\Some(a))\vee u=a)
    \wedge \beta(a) \sqsubseteq l \\

    \{ {\tt env}[u\mapsto \Some(a)] \mid \beta(a) \sqsubseteq l \} \\
    \hspace{5mm} 
    \hbox{if }
    V = [u] \wedge u\in\mathcal{X}\wedge {\tt env}[u]=\None\\

\emptyset \ \ \ \hbox{otherwise}
\end{array}\right.
$$
The function is parametrized with $\beta : \U \rightarrow \L$, defined
in Section \ref{sec:alfp-lat}. It performs a unification of an lattice
term $V$ with an element $l \in \L$ in the environment {\tt env}. In
the case when the unification succeeds the set of unified environments
is returned, otherwise the function returns empty set.

The other important operation on the partial environment is given by
the function {\sc unifiable}. The function when applied to {\tt env}
and a tuple $(\vec{u};V)$, returns a set of tuples for which {\sc
  unify} would succeed. The function is defined by means of two
auxiliary functions, formally we have
$$
\textsc{unifiable}(\texttt{env},(\vec{u};V)) =
(\textsc{unifiable}_\textsc{U}(\texttt{env},\vec{u});
\textsc{unifiable}_\textsc{L}(\texttt{env},V))
$$
where
$$
\textsc{unifiable}_\textsc{U}(\texttt{env},u) =
\left\{
\begin{array}{cl}
\{ a \} & \hbox{if } (u\in\mathcal{X}\wedge{\tt env}[u]=\Some(a))\vee
u=a\\

\U & \hbox{if } u\in\mathcal{X}\wedge{\tt env}[u]=\None\\
\end{array}
\right.
$$
and
$$
\begin{array}{ll}
\textsc{unifiable}_\textsc{L}(\texttt{env},V) =
\left\{
\begin{array}{ll}

l & \hbox{if } V\in\mathcal{Y}\wedge{\tt env}[V]=\Some(l)\\

\top & \hbox{if } V\in\mathcal{Y}\wedge{\tt env}[V]=\None\\

\beta(a) & \hbox{if } V=[u] \wedge (u=a \vee \\
& \hspace{3mm} (u\in\mathcal{X}\wedge{\tt env}[u]=\Some(a)))\\

\bigsqcup \{ \beta(a) \mid a \in \U \} & \hbox{if } V=[u] \wedge
u\in\mathcal{X}\wedge\\
& \hspace{3mm} {\tt env}[u]=\None\\

\sem{f}(l) & \hbox{if }
V=f(\vec{V}) \wedge \\
& \hspace{3mm} l = \textsc{unifiable}_\textsc{L}(\texttt{env},\vec{V}) \\

\end{array}
\right.
\end{array}
$$
Both auxiliary funcions are extended to $k$-tuples in a
straightforward way.

The global data structure \texttt{result}, which is updated
incrementally during computations, is represented as a mapping from
predicate names to the prefix trees that for each predicate $R$ record
the tuples currently known to belong to $R$. There are three main
operations on the data structure \texttt{result}: the operation
\texttt{result.}\textsc{has} checks whether a given tuple is
associated with a given predicate, the operation
\texttt{result.}\textsc{sub} returns a list of the tuples associated
with a given predicate and the operation \texttt{result.}\textsc{add}
adds a tuple to the interpretation of a given predicate.

Since $\varrho$ is updated as the algorithm progresses, it may
happen that a query $R(\vec v; V)$ inside a precondition fails to be
satisfied at the given point in time, but may hold in the future when
a new tuple $(\vec a; l)$ is added to the interpretation of $R$. If we
are not careful we may lose the consequences that adding $(\vec a; l)$
to $R$ will have on the contents of other predicates. This gives rise
to the data structure \texttt{infl} that records computations that
have to be resumed for the new tuples; these future computations are
called \emph{consumers}. The \texttt{infl} data structure is also
represented as a mapping from the predicate names to prefix trees that
for each predicate $R$ record consumers that have to be resumed when
the interpretation of $R$ is updated. There are two main operations on
the data structure {\tt infl}: the operation
\texttt{infl.}\textsc{register} that adds a new consumer for a given
predicate and \texttt{infl.}\textsc{consumers} that returns all the
consumers currently associated with a given predicate.

In the algorithm, we have one function for each of the three syntactic
categories. The function \textsc{solve} takes a \emph{clause sequence}
as input and calls the function \textsc{execute} on each of the
individual clauses
\[
\textsc{solve}(cl_1,\ldots,cl_s) =\textsc{execute}(cl_1)[\ ]; \ldots;
\textsc{execute}(cl_s)[\ ]
\]
where we write [\ ] for the empty environment reflecting that we have
no free variables in the clause sequences.

Let us now turn to the description of the function
\textsc{execute}. The function takes a \emph{clause} $cl$ as a
parameter and a representation {\tt env} of the interpretation of the
variables. We have one case for each of the forms of $cl$; the pseudo
code is given in Figure \ref{figure:function-execute}.
\begin{figure}
\centering
\begin{minipage}{.77\textwidth}
\begin{algorithmic}
\State \Call{execute}{$R(\vec{v};V)}$\texttt{env} =
\State \indent \textbf{let} \Call{iterFun}{} $(\vec{a};l)$ =
\State \indent \indent \textbf{match} \texttt{result.}\Call{has}{$R, (\vec{a};l)$} \textbf{with}
\State \indent \indent \indent $|$ $true$ $\rightarrow$ ()
\State \indent \indent \indent $|$ $false$ $\rightarrow$
\State \indent \indent \indent \indent \texttt{result.}\Call{add}{$R,(\vec{a};l)$}
\State \indent \indent \indent \indent \Call{iter}{} (\textbf{fun} $f$ $\rightarrow$ $f$ $(\vec{a};l)$) (\texttt{infl.}\Call{consumers}{} $R)$
\State \indent \textbf{in} \Call{iter}{} \Call{iterFun}{} (\Call{unifiable}{\texttt{env},$(\vec{v};V)$})
\end{algorithmic}
\begin{algorithmic}
\State$\textsc{execute}({\bf 1}){\tt env} = ()$
\end{algorithmic}
\begin{algorithmic}
  \State$\textsc{execute}(cl_1\wedge cl_2){\tt env} =
  \textsc{execute}(cl_1){\tt env}; \textsc{execute}(cl_2){\tt env}$
\end{algorithmic}
\begin{algorithmic}
  \State$\textsc{execute}(pre\Rightarrow cl){\tt env} =
  \textsc{check}(pre,\textsc{execute}(cl)){\tt env}$
\end{algorithmic}
\begin{algorithmic}
  \State$\textsc{execute}(\forall x: cl){\tt env} =
  {\textsc{execute}}(cl)({\tt env}[x \mapsto \None])$
\end{algorithmic}
\end{minipage}
\caption[]{The \textsc{execute} function.}
\label{figure:function-execute}
\end{figure}
Let us explain the case of an assertion first. The algorithm uses the
auxiliary function \textsc{iter}, which applies the function
\textsc{iterFun} to each element of the list of tuples that can be
unified with the argument $(\vec{v};V)$. Given a tuple $(\vec{a};l)$,
the function \textsc{iterFun} adds the tuple to the interpretation of
$R$ stored in \texttt{result} if it is not already present. If the
\textsc{add} operation succeeds, we first create a list of all the
consumers currently registered for predicate $R$ by calling the
function \texttt{infl.}\textsc{consumers}. Thereafter, we resume the
computations by iterating over the list of consumers and calling
corresponding continuations.
The cases of always true clause, {\bf 1}, is straightforward; the
function simply returns the unit, without performing any other
actions.
In the case of the conjunction of clauses the algorithm calls the
\textsc{execute} function for both conjuncts and the current
environment \texttt{env}.
In the case of implication we make use of the function \textsc{check}
that in addition to the precondition and the environment also takes
the continuation $\textsc{execute}(cl)$ as an argument.
In the case of universal quantification, we simply extend the
environment to record that the value of the new variable is unknown
and then we recurse. The case of universal quantification over a
variable $Y \in \mathcal{Y}$ is exactly the same and hence omitted.

Now, let us present the function \textsc{check}. It takes a
\emph{precondition}, a continuation and an environment as
parameters. The pseudo code is given in Figure
\ref{fig:function-check}.
\begin{figure}
\centering
\begin{minipage}{.95\textwidth}
\begin{algorithmic}
\State \Call{check}{$R(\vec{v};V),next}$\texttt{env} =
\State \indent \textbf{let} \Call{consumer}{} $(\vec{a};l)$ =
\State \indent \indent \textbf{match} \Call{unify}{\texttt{env}$, (\vec{v};V), (\vec{a};l)$} \textbf{with}
\State \indent \indent \indent $|$ fail $\rightarrow$ ()
\State \indent \indent \indent $|$ \texttt{envs} $\rightarrow$
\Call{iter}{} $next$ \texttt{envs}
\State \indent \textbf{in} \texttt{infl.}\Call{register}{}($R$,\Call{consumer}{}); \Call{iter}{} \Call{consumer}{} (\texttt{result.}\Call{sub}{} $R$)
\end{algorithmic}
\begin{algorithmic}
\State \Call{check}{$\neg R(\vec{v};V),next}$\texttt{env} =
\State \indent \textbf{let} \Call{iterFun}{} $(\vec{a};l)$ =
\State \indent \indent \textbf{match} \texttt{result.}\Call{has}{$R, (\vec{a};l)$} \textbf{with}
\State \indent \indent \indent $|$ $true$ $\rightarrow$ ()
\State \indent \indent \indent $|$ $false$ $\rightarrow$ \Call{iter}{}
$next$ (\Call{unify}{\texttt{env}$, (\vec{v};V), (\vec{a};l)$})
\State \indent \textbf{in} \Call{iter}{} \Call{iterFun}{} (\Call{unifiable}{\texttt{env}$, (\vec{v};V)$})
\end{algorithmic}
\begin{algorithmic}
\State \Call{check}{$Y(x),next}$\texttt{env} =
\State \indent \textbf{let} \texttt{env'} = \textbf{if}
\texttt{env}$(Y) = \Some(l)$ \textbf{then} \texttt{env} \textbf{else} \texttt{env}$[Y
\mapsto \top]$
\State \indent \textbf{in} \textbf{let} \textsc{f} a = $\ifit \Some(\beta(a))
\sqsubseteq \
${\tt env'}$(Y) \then next \ ${\tt env'}$[x \mapsto a] \elseit ()$
\State \indent \textbf{in} \textbf{match} \texttt{env'}(x) \textbf{with}
\State \indent \indent $|$ $\Some(a)$ $\rightarrow$ \textsc{f} a
\State \indent \indent $|$ $\None$ $\rightarrow$ \Call{iter}{} \textsc{f} $U$
\end{algorithmic}
\begin{algorithmic}
  \State$\textsc{check}(pre_1\wedge pre_2,next){\tt env} =
  \textsc{check}(pre_1,\textsc{check}(pre_2,next)){\tt env}$
\end{algorithmic}
\begin{algorithmic}
\State \Call{check}{$pre_{1} \vee pre_{2},next}$\texttt{env} = \Call{check}{$pre_1,next$}\texttt{env}; \Call{check}{$pre_2,next$}\texttt{env}
\end{algorithmic}
\begin{algorithmic}
\State \Call{check}{$\exists x: pre,next}$\texttt{env} =
\Call{check}{$pre,next\ \circ\ $(\textsc{remove} $x$)}(\texttt{env}$[x \mapsto \None]$)
\end{algorithmic}
\end{minipage}
\caption[]{The \textsc{check} function.}
\label{fig:function-check}
\end{figure}
In the case of positive queries we first ensure that the consumer is
registered in {\tt infl}, by calling function \textsc{register}, so
that future tuples associated with $R$ will be processed. Thereafter,
the function inspects the data structure \texttt{result} to obtain the
list of tuples associated with the predicate $R$. Then, the auxiliary
function \textsc{consumer} unifies $(\vec{v};V)$ with each tuple; and
if the operation succeeds, the continuation $next$ is invoked on each
of the updated new environments in the returned set \texttt{envs}.
In the case of negated query, the algorithm first computes the tuples
unifiable with $(\vec{v};V)$ in the environment \texttt{env}. Then,
for each tuple it checks whether the tuple is already in $R$ and if
not, the tuple is unified with $(\vec{v};V)$ to produce set of new
environments. Thereafter, the continuation $next$ is evaluated in each
of the environments contained in the returned set. Notice that in the
case of negative queries we do not register a consumer for the
relation $R$. This is because the stratification condition introduced
in Definition \ref{def:stratification} ensures that the relation is
fully evaluated before it is queried negatively. Thus, there is no
need to register future computations since the interpretation of $R$
will not change.
Now, let us consider function \textsc{check} in the case of $Y(x)$,
where $x \in \mathcal{X}$. The function begins with creating an
environment {\tt env'} that is exactly as {\tt env} except that the
binding for the variable $Y$ is set to $\top$ in the case $Y$ is
undefined in {\tt env}. Then, we define an auxiliary function that
checks whether {\tt env'}$(Y)$ over-approximates the abstraction of an
argument $a$, denoted by $\beta(a)$, and if so the continuation is
called in the environment {\tt env'}$[x \mapsto a]$. Finally, the
function checks the binding for the variable $x$ in the environment
{\tt env'} and if it is bound to $\Some(a)$ the function \textsc{f}
applied to $a$ is called. Otherwise, the function \textsc{f} is called
for each element of the universe, using the \textsc{iter}
function. The case of $Y(a)$, where $a \in \U$ is essentially the same
as the case explained above, except that we do not have to handle the
case when $x \in \mathcal{X}$ is undefined in {\tt env}.
For conjunction of preconditions we exploit a continuation passing
programming style. More precisely, we call the \textsc{check} function
for the precondition $pre_1$, and as a continuation we pass a call to
the \textsc{check} function partially applied to the precondition
$pre_2$ and the continuation $next$.
In the case of disjunction of preconditions the function simply checks
preconditions $pre_1$ and $pre_2$ respectively in the current
environment \texttt{env}. In order to be efficient we use memoization;
this means that if both checks yield the same bindings of variables,
the second check does not need to consider the continuation, as it has
already been done.
The algorithm for existential quantification checks the precondition
$pre$ in the environment extended with the quantified variable. The
continuation that is passed is a composition of functions $next$ and
\textsc{remove} $x$, where the function \textsc{remove} removes
variable passed as a first argument from the environment passed as a
second argument. In order to be efficient we again use a memoization
to avoid redundant computations. The case of existential
quantification over a variable $Y \in \mathcal{Y}$ is exactly the same
and hence omitted.

%% file: conclusions.tex
In the paper we introduced the LLFP logic, which is an expressive
formalism for specifying static analysis problems. It lifts the
limitation of logics such as Datalog and ALFP by allowing
interpretation over complete lattices satisfying Ascending Chain
Condition. Thanks to the declarative style, the analysis
specifications are easy to analyse for their correctness. 

We established a Moore Family result that guarantees that there always
is a unique best solution for the LLFP formulae. More generally this
ensures that the approach taken falls within the general Abstract
Interpretation framework. We also developed a state-of-the-art solving
algorithm for LLFP, which is a continuation passing style algorithm,
which represents relations as prefix trees. We showed that the logic
and the associated solver can be used for rapid prototyping of
sophisticated static analyses by presenting the formulation of
interval analysis.

As a future work we plan to implement a front-end to automatically
extract analysis relations from program source code, and perform
experiments on real-world programs in order to evaluate the
performance of the LLFP solver. Furthermore, we would like to lift the
Ascending Chain Condition and use e.g. \textit{widening operator}
\cite{DBLP:journals/jlp/CousotC92,DBLP:conf/plilp/CousotC92} in order
to ensure termination of the least fixed point computation.

%% file: proof-partial-order-alfp-lat.tex
These appendices are not intended for publication and references to
them will be removed in the final version.

\section{Proof of Lemma \ref{lemma:partial-order-alfp-lat}}
\label{proof:lemma:partial-order-alfp-lat}

\begin{proof}

{\bf Reflexivity} $\forall \varrho \in \Delta: \varrho \leqR \varrho$.

\noindent To show that $\varrho \leqR \varrho$ let us take $j = s$. If
$\rank(R)<j$ then $\varrho(R)=\varrho(R)$ as required. Otherwise if
$\rank(R)=j$ then from $\varrho(R)=\varrho(R)$ we get $\varrho(R)
\sqsubseteq \varrho(R)$. Thus we get the required $\varrho \leqR
\varrho$.

\noindent {\bf Transitivity} $\forall \varrho_1, \varrho_2, \varrho_3
\in \Delta: \varrho_1 \leqR \varrho_2 \wedge \varrho_2
\leqR \varrho_3 \Rightarrow \varrho_1 \leqR \varrho_3$.

\noindent Let us assume that $\varrho_1 \leqR \varrho_2 \wedge
\varrho_2 \leqR \varrho_3$. From $\varrho_i \leqR
\varrho_{i+1}$ we have $j_i$ such that conditions
(\ref{itm:alfp-lat-ord-rank-less})--(\ref{itm:alfp-lat-ord-rank-s})
are fulfilled for $i=1,2$. Let us take $j$ to be the minimum of $j_1$
and $j_2$. Now we need to verify that conditions
(\ref{itm:alfp-lat-ord-rank-less})--(\ref{itm:alfp-lat-ord-rank-s})
hold for $j$. If $\rank(R)<j$ we have $\varrho_1(R) = \varrho_2(R)$ and
$\varrho_2(R) = \varrho_3(R)$. It follows that $\varrho_1(R) =
\varrho_3(R)$, hence (\ref{itm:alfp-lat-ord-rank-less}) holds. Now let
us assume that $\rank(R)=j$. We have $\varrho_1(R) \sqsubseteq
\varrho_2(R)$ and $\varrho_2(R) \sqsubseteq \varrho_3(R)$ and from
transitivity of $\sqsubseteq$ we get $\varrho_1(R) \sqsubseteq
\varrho_3(R)$, which gives (\ref{itm:alfp-lat-ord-rank-eq}). Let us
now assume that $j \neq s$, hence $\varrho_i(R) \sqsubset
\varrho_{i+1}(R)$ for some $R \in \R$ and $i=1,2$. Without loss of
generality let us assume that $\varrho_1(R) \sqsubset
\varrho_2(R)$. We have $\varrho_1(R) \sqsubset \varrho_2(R)$ and
$\varrho_2(R) \sqsubseteq \varrho_3(R)$, hence $\varrho_1(R)
\sqsubset \varrho_3(R)$, and (\ref{itm:alfp-lat-ord-rank-s}) holds.

\noindent {\bf Anti-symmetry} $\forall \varrho_1, \varrho_2 \in \Delta: \varrho_1
\leqR \varrho_2 \wedge \varrho_2 \leqR \varrho_1
\Rightarrow \varrho_1 = \varrho_2$.

\noindent Let us assume $\varrho_1 \leqR \varrho_2$ and
$\varrho_2 \leqR \varrho_1$. Let $j$ be minimal such that
$\rank(R)=j$ and $\varrho_1(R) \neq \varrho_2(R)$ for some $R \in
\R$. Then, since $\rank(R)=j$, we have $\varrho_1(R) \sqsubseteq
\varrho_2(R)$ and $\varrho_2(R) \sqsubseteq \varrho_1(R)$.  Hence
$\varrho_1(R) = \varrho_2(R)$ which is a contradiction. Thus it must
be the case that $\varrho_1(R) = \varrho_2(R)$ for all $R \in \R$.
\qed
\end{proof}

%% file: proof-complete-lattice-alfp-lat.tex
\section{Proof of Lemma \ref{lemma:complete-lattice-alfp-lat}}
\label{proof:lemma:complete-lattice-alfp-lat}

\begin{proof}
  First we prove that $\glb M$ is a lower bound of $M$; that is $\glb
  M \leqR {\varrho}$ for all ${\varrho} \in M$. Let $j$ be maximum
  such that ${\varrho} \in M_j$; since $M=M_0$ and $M_j \supseteq
  M_{j+1}$ clearly such $j$ exists. From definition of $M_j$ it
  follows that $(\glb M)(R)={\varrho}(R)$ for all $R$ with
  $\rank(R)<j$; hence (\ref{itm:alfp-lat-ord-rank-less}) holds. If
  $\rank(R)=j$ we have $(\glb M)(R)= \lambda \vec{a} . \bigsqcap \{
  {\varrho}'(R)(\vec{a}) \mid {\varrho}' \in M_j \} \sqsubseteq
  {\varrho}(R)$ showing that (\ref{itm:alfp-lat-ord-rank-eq})
  holds. Finally let us assume that $j \neq s$; we need to show that
  there is some $R$ with $\rank(R)=j$ such that $(\glb
  M)(R)\sqsubset{\varrho}(R)$. Since we know that $j$ is maximum such
  that ${\varrho} \in M_j$, it follows that ${\varrho} \notin
  M_{j+1}$, hence there is a relation $R$ with $\rank(R)=j$ such that
  $(\glb M)(R)\sqsubset{\varrho}(R)$; thus
  (\ref{itm:alfp-lat-ord-rank-s}) holds.

  \noindent Now we need to show that $\glb M$ is the greatest lower
  bound. Let us assume that ${\varrho}' \leqR {\varrho}$ for all
  ${\varrho} \in M$, and let us show that ${\varrho}' \leqR \glb
  M$. If ${\varrho}' = \glb M$ the result holds vacuously, hence let
  us assume ${\varrho}' \neq \glb M$. Then there exists a minimal $j$
  such that $(\glb M)(R)\neq{\varrho}'(R)$ for some $R$ with
  $\rank(R)=j$. Let us first consider $R$ such that $\rank(R)<j$. By
  our choice of $j$ we have $(\glb M)(R) = {\varrho}'(R)$ hence
  (\ref{itm:alfp-lat-ord-rank-less}) holds. Next assume that
  $\rank(R)=j$. Since we assumed that ${\varrho}' \leqR {\varrho}$ for
  all ${\varrho} \in M$ and $M_j \subseteq M$, it follows that
  ${\varrho}'(R) \sqsubseteq {\varrho}(R)$ for all ${\varrho} \in
  M_j$. Thus we have ${\varrho}'(R) \sqsubseteq \lambda \vec{a}
  . \bigsqcap \{ {\varrho}(R)(\vec{a}) \mid {\varrho} \in M_j
  \}$. Since $(\glb M)(R) = \lambda \vec{a} . \bigsqcap \{
  {\varrho}(R)(\vec{a}) \mid {\varrho} \in M_j \}$, we have
  ${\varrho}'(R) \sqsubseteq (\glb M)(R)$ which proves
  (\ref{itm:alfp-lat-ord-rank-eq}). Finally since we assumed that
  ${\varrho}'(R)\neq(\glb M)(R)$ for some $R$ with $\rank(R)=j$,
  it follows that (\ref{itm:alfp-lat-ord-rank-s}) holds. Thus we
  proved that ${\varrho}' \leqR \glb M$.  \qed

\end{proof}

%% file: proof-moore-family.tex
\section{Proof of Proposition \ref{prop:moore-family-alfp-lat}}
\label{proof:prop:moore-family-alfp-lat}

In order to prove Proposition \ref{prop:moore-family-alfp-lat} we
first state and prove two auxiliary lemmas.

\begin{lemma}
\label{lemma:glb-pre-alfp-lat}
If $\varrho = \glb M$, $pre$ occurs in $cl_j$ and
$(\varrho, \varsigma) \models_{\beta} pre$ then also
$({\varrho}', {\varsigma}) \models_{\beta} pre$ for all
${\varrho}' \in M_j$.
\end{lemma}

\begin{proof}
  We proceed by induction on $j$ and in each case perform a structural
  induction on the form of the precondition $pre$ occurring in
  $cl_j$.\newline \textbf{Case: }$pre=R(\vec{u};V)$ \newline Let us
  take $\varrho = \glb M$ and assume that
\[
({\varrho}, {\varsigma}) \models_{\beta} R(\vec{u};V)
\]
From Table \ref{ALFPStarsemantics} we have:
\[
\varrho(R)(\varsigma(\vec u)) \sqsupseteq \varsigma(V)
\]
Depending on the rank of $R$ we have two cases. If $\rank(R)=j$ then
$\varrho(R) = \lambda \vec{a} . \bigsqcap \{ {\varrho}'(R)(\vec{a}) \mid {\varrho}' \in M_j \}$
and hence we have
\[
\bigsqcap \{ {\varrho}'(R)(\varsigma(\vec u)) \mid {\varrho}' \in M_j
\} \sqsupseteq \varsigma(V)
\]
It follows that for all ${\varrho}' \in M_j$
\[ 
{\varrho}'(R)(\varsigma(\vec u)) \sqsupseteq \varsigma(V)
\]
Now if $\rank(R)<j$ then ${\varrho}(R) = {\varrho}'(R)$ for all
${\varrho}' \in M_j$ hence we have that for all
${\varrho}' \in M_j$
\[
{\varrho}'(R)(\varsigma(\vec u)) \sqsupseteq \varsigma(V)
\]
which according to Table \ref{ALFPStarsemantics} is equivalent to
\[
\forall {\varrho}' \in M_j:
({\varrho}', {\varsigma}) \models_{\beta} R(\vec{u};V)
\]
which was required and finishes the case.\newline
\textbf{Case: }$pre=Y(u)$ \newline Let us
  take $\varrho = \glb M$ and assume that
\[
({\varrho}, {\varsigma}) \models_{\beta} Y(u)
\]
According to the semantics of LLFP in Table \ref{ALFPStarsemantics}
we have
\[
\beta(\varsigma(u)) \sqsubseteq \varsigma(Y)
\]
It follows that 
\[
\forall {\varrho}' \in M_j: \beta(\varsigma(u)) \sqsubseteq \varsigma(Y)
\]
which according to the semantics of LLFP in Table
\ref{ALFPStarsemantics} is equivalent to
\[
\forall {\varrho}' \in M_j: ({\varrho}', {\varsigma})
\models_{\beta} Y(u)
\]
which was required and finishes the case.\newline
\textbf{Case: }$pre=\neg R(\vec{u};V)$ \newline Let us take
$\varrho = \glb M$ and assume that
\[
({\varrho}, {\varsigma}) \models_{\beta} \neg R(\vec{u};V)
\]
From Table \ref{ALFPStarsemantics} we have:
\[
\complement(\varrho(R)(\varsigma(\vec{u}))) \sqsupseteq \varsigma(V)
\]
Since $\rank(R)<j$ then we know that ${\varrho}(R) = {\varrho}'(R)$ for all
${\varrho}' \in M_j$ hence we have that
\[
\forall {\varrho}' \in M_j: \complement(\varrho(R)(\varsigma(\vec{u}))) \sqsupseteq \varsigma(V)
\]
Which according to Table \ref{ALFPStarsemantics} is equivalent to
\[
\forall {\varrho}' \in M_j:
({\varrho}', {\varsigma}) \models_{\beta} \neg R(\vec{u};V)
\]
which was required and finishes the case.\newline
\textbf{Case: }$pre=pre_1 \wedge pre_2$ \newline Let us take
$\varrho = \glb M$ and assume that
\[
({\varrho}, {\varsigma}) \models_{\beta} pre_1
\wedge pre_2
\]
According to Table \ref{ALFPStarsemantics} we have
\[
({\varrho}, {\varsigma}) \models_{\beta} pre_1
\]
and 
\[
({\varrho}, {\varsigma}) \models_{\beta} pre_2
\]
From the induction hypothesis we get that for all ${\varrho}' \in
M_j$
\[
({\varrho}', {\varsigma}) \models_{\beta} pre_1
\]
and 
\[
({\varrho}', {\varsigma}) \models_{\beta} pre_2
\]
It follows that for all ${\varrho}' \in
M_j$
\[
({\varrho}', {\varsigma}) \models_{\beta} pre_1
\wedge pre_2
\]
which was required and finishes the case.\newline
\textbf{Case: }$pre=pre_1 \vee pre_2$ \newline Let us take
$\varrho = \glb M$ and assume that
\[
({\varrho}, {\varsigma}) \models_{\beta} pre_1
\vee pre_2
\]
According to Table \ref{ALFPStarsemantics} we have
\[
({\varrho}, {\varsigma}) \models_{\beta} pre_1
\]
or 
\[
({\varrho}, {\varsigma}) \models_{\beta} pre_2
\]
From the induction hypothesis we get that for all ${\varrho}' \in
M_j$
\[
({\varrho}', {\varsigma}) \models_{\beta} pre_1
\]
or 
\[
({\varrho}', {\varsigma}) \models_{\beta} pre_2
\]
It follows that for all ${\varrho}' \in
M_j$
\[
({\varrho}', {\varsigma}) \models_{\beta} pre_1
\vee pre_2
\]
which was required and finishes the case.\newline \textbf{Case: }$pre=
\exists x : pre'$\newline Let us take $\varrho = \glb M$ and
assume that
\[
({\varrho}, {\varsigma}) \models_{\beta} \exists x
: pre'
\]
According to Table \ref{ALFPStarsemantics} we have
\[
\exists a \in \U : ({\varrho}, {\varsigma}[x \mapsto a]) \models_{\beta} pre'
\]
From the induction hypothesis we get that for all ${\varrho}' \in M_j$
\[
\exists a \in \U : ({\varrho}', {\varsigma}[x \mapsto a]) \models_{\beta} pre'
\]
It follows from Table \ref{ALFPStarsemantics} that for all ${\varrho}' \in M_j$
\[
({\varrho}', {\varsigma}) \models_{\beta} \exists x
: pre'
\]
which was required and finishes the case.\newline \textbf{Case: }$pre=
\exists Y : pre'$\newline Let us take $\varrho = \glb M$ and
assume that
\[
({\varrho}, {\varsigma}) \models_{\beta} \exists Y
: pre'
\]
According to Table \ref{ALFPStarsemantics} we have
\[
\exists l \in \Lnb : ({\varrho}, {\varsigma}[Y \mapsto l]) \models_{\beta} pre'
\]
From the induction hypothesis we get that for all ${\varrho}' \in M_j$
\[
\exists l \in \Lnb : ({\varrho}', {\varsigma}[Y \mapsto l]) \models_{\beta} pre'
\]
It follows from Table \ref{ALFPStarsemantics} that for all ${\varrho}' \in M_j$
\[
({\varrho}', {\varsigma}) \models_{\beta} \exists Y
: pre'
\]
which was required and finishes the case.\qed
\end{proof}


\begin{lemma}
\label{lemma:glb-cl-alfp-lat}
If $\varrho=\glb M$ and $({\varrho}', \zeta, \varsigma)
\models_{\beta} cl_j$ for all ${\varrho}' \in M$ then $(\varrho,
\zeta, \varsigma) \models_{\beta} cl_j$.
\end{lemma}
\begin{proof}
  We proceed by induction on $j$ and in each case perform a structural
  induction on the form of the clause occurring in $cl_j$.\newline
\textbf{Case: }$cl_j=R(\vec{u};V)$\newline Assume that for all
  ${\varrho}' \in M$
  \[
  ({\varrho}', \zeta, {\varsigma}) \models_{\beta} R(\vec{u};V)
  \]
  From the semantics of LLFP we have that for all ${\varrho}'
  \in M$
 \[ 
 {\varrho}'(R)(\sem{\vec u}(\zeta, \varsigma)) \sqsupseteq
 \sem{V}(\zeta, \varsigma)
  \]
  It follows that:
  \[
  \bigsqcap \{ {\varrho}'(R)(\sem{\vec u}(\zeta, \varsigma)) \mid
  {\varrho}' \in M \} \sqsupseteq \sem{V}(\zeta, \varsigma)
  \]
  Since $M_j \subseteq M$, we have:
   \[
  \bigsqcap \{ {\varrho}'(R)(\sem{\vec u}(\zeta, \varsigma)) \mid
  {\varrho}' \in M_j \} \sqsupseteq \sem{V}(\zeta, \varsigma)
  \]
  We know that $\rank(R)=j$; hence $\varrho(R) = \lambda \vec{a}
  . \bigsqcap \{ {\varrho}'(R)(\vec{a}) \mid {\varrho}' \in M_j \}$;
  thus
\[
{\varrho}(R)(\sem{\vec u}(\zeta, \varsigma)) = \bigsqcap \{
{\varrho}'(R)(\sem{\vec u}(\zeta, \varsigma)) \mid {\varrho}' \in M_j
\} \sqsupseteq \sem{V}(\zeta, \varsigma)
\]
Which according to Table \ref{ALFPStarsemantics} is equivalent to
  \[
  ({\varrho}, \zeta, \varsigma) \models_{\beta} R(\vec{u};V)
  \]
\textbf{Case: }$cl_j= cl_1 \wedge cl_2$ \newline Assume that for all
  ${\varrho}' \in M$:
\[
({\varrho}', \zeta, {\varsigma}) \models_{\beta} cl_1 \wedge cl_2
\]
From Table \ref{ALFPStarsemantics} it is equivalent to
\[
({\varrho}', \zeta, {\varsigma}) \models_{\beta} cl_1 \text{ and }
({\varrho}', \zeta, {\varsigma}) \models_{\beta} cl_2
\]
The induction hypothesis gives that
\[
({\varrho}, \zeta, {\varsigma}) \models_{\beta} cl_1 \text{ and }
({\varrho}, \zeta, {\varsigma}) \models_{\beta} cl_2
\]
Which according to Table \ref{ALFPStarsemantics} is equivalent to
\[
({\varrho}, \zeta, {\varsigma}) \models_{\beta} cl_1 \wedge cl_2
\]
and finishes the case.\newline
  \textbf{Case: }$cl_j=pre \Rightarrow cl$ \newline Assume that for all
  ${\varrho}' \in M$:
  \begin{equation}
    ({\varrho}', \zeta, {\varsigma}) \models_{\beta} pre \Rightarrow cl \label{eq:assumption-imply}
  \end{equation} 
  We have two cases. In the first one $(\varrho, \varsigma) \models_{\beta} pre$ is $false$, hence $(\varrho,
\varsigma, \zeta) \models_{\beta} pre \Rightarrow cl$ holds
trivially. In the second case let us assume:
  \begin{equation}
    (\varrho, \varsigma) \models_{\beta} pre \label{eq:assumption-pre-imply}
  \end{equation} 
  Lemma \ref{lemma:glb-pre-alfp-lat} gives that for all ${\varrho}'
  \in M_j$
  \[
  ({\varrho}', {\varsigma}) \models_{\beta} pre
  \]
  From \eqref{eq:assumption-imply} we have that for all
  ${\varrho}' \in M_j$
  \[
  ({\varrho}', \zeta, {\varsigma}) \models_{\beta} cl
  \]
  and the induction hypothesis gives:
  \[
  ({\varrho}, \zeta, \varsigma) \models_{\beta} cl
  \]
  Hence from \eqref{eq:assumption-pre-imply} we get:
  \[
  ({\varrho}, \zeta, \varsigma) \models_{\beta} pre \Rightarrow cl
  \]
  which was required and finishes the case. \newline \textbf{Case: }
  $cl_j=\forall x: cl$\newline Assume that for all
  ${\varrho}' \in M$
  \[
  ({\varrho}', \zeta, {\varsigma}) \models_{\beta} \forall x: cl
  \]
  From Table \ref{ALFPStarsemantics} we have that for all
  ${\varrho}' \in M$ and for all $a \in \U$
  \[
    ({\varrho}', \zeta, {\varsigma}[x \mapsto a]) \models_{\beta} cl
  \]
  Thus from the induction hypothesis we get that for all $a \in \U$
  \[
  ({\varrho}, \zeta, \varsigma[x \mapsto a]) \models_{\beta} cl
  \]
  According to Table \ref{ALFPStarsemantics} it is equivalent to
  \[
  ({\varrho},\zeta, \varsigma) \models_{\beta} \forall x: cl
  \]
  which was required and finishes the case.\newline \textbf{Case: }
  $cl=\forall Y: cl$\newline Assume that for all ${\varrho}' \in
  M$
  \[
({\varrho}', \zeta, {\varsigma}) \models_{\beta} \forall Y: cl
  \]
  From Table \ref{ALFPStarsemantics} we have that ${\varrho}' \in
  M$
  \[
  \forall l \in \Lnb: ({\varrho}', \zeta,{\varsigma}[ Y \mapsto l ]) \models_{\beta} cl
  \]
  Thus from the induction hypothesis we get that
  \[
\forall l \in \Lnb:
  ({\varrho}, \zeta,\varsigma[Y \mapsto l]) \models_{\beta} cl
  \]
  According to Table \ref{ALFPStarsemantics} it is equivalent to
  \[
  ({\varrho}, \zeta,\varsigma) \models_{\beta} \forall Y: cl
  \]
  which was required and finishes the case.
\end{proof}

\noindent{\bf Proposition \ref{prop:moore-family-alfp-lat}.}
Assume $cls$ is a stratified LLFP clause sequence, $\varsigma_0$ and
$\zeta_0$ are interpretations of free variables and function symbols
in $cls$, respectively. Furthermore, $\varrho_0$ is an interpretation
of all relations of rank 0. Then $\{ \varrho \mid (\varrho, \zeta_0,
\varsigma_0) \models_{\beta} cls \wedge \forall R: \rank(R) =
0 \Rightarrow \varrho_0(R) \sqsubseteq \varrho(R) \}$ is a Moore
family.

\begin{proof}
The result follows from Lemma \ref{lemma:glb-cl-alfp-lat}.
\qed
\end{proof}